\documentclass{LMCS}
\usepackage{amsmath,amssymb,amsfonts,amsfonts,stmaryrd}
\usepackage{enumerate,hyperref}
\usepackage{bbm}
\usepackage{bussproofs}
\usepackage{graphics}
\usepackage[all]{xy}
{\bfseries}{\upshape}

\newenvironment{conditionsiii}
{%
	\begin{list}{\rm (\roman{enumi})}%
	{\noindent%
		\usecounter{enumi}%
		\setlength{\topsep}{2pt}%
		\setlength{\partopsep}{0pt}%
		\setlength{\itemsep}{2pt}%
		\setlength{\parsep}{0pt}%
		\setlength{\leftmargin}{2.5em}%
		\setlength{\labelwidth}{1.5em}%
		\setlength{\labelsep}{0.5em}%
		\setlength{\listparindent}{0pt}%
		\setlength{\itemindent}{0pt}%
	}%
}%
{\end{list}}%

\newcommand{\nto}{\rightarrowtriangle}

\newcommand{\cl}[1]{\mathcal{#1}}

\newcommand{\imp}{\mbox{$\,{\rightarrow}\,$}}

\newcommand{\Fdd}{F_{DL}(\imp(D\times D))}

\newcommand{\Set}{{\bf Set}}

\def\doi{7 (2:9) 2011}
\lmcsheading%
{\doi}
{1--24}
{}
{}
{Oct.~\phantom05, 2010}
{May\phantom.~17, 2011}
{}

\begin{document}

\title[Finitely generated free Heyting algebras via Birkhoff duality and coalgebra]{Finitely generated free Heyting algebras
via Birkhoff duality and coalgebra}

\author[N.~Bezhanishvili]{Nick Bezhanishvili\rsuper a}	
\address{{\lsuper a}Department of Computing, Imperial College London}	
\email{nbezhani@doc.ic.ac.uk}  
\thanks{{\lsuper a}Partially supported by EPSRC EP/F032102/1 and GNSF/ST08/3-397}

\author[M.~Gehrke]{Mai Gehrke\rsuper b}	
\address{{\lsuper b}IMAPP, Radboud Universiteit Nijmegen, the Netherlands}	
\email{mgehrke@math.ru.nl}  
\thanks{{\lsuper b}Partially supported by EPSRC EP/E029329/1 and EP/F016662/1}


\begin{abstract}
Algebras axiomatized entirely by rank 1 axioms are algebras for a functor and 
thus the free algebras can be obtained by a direct limit process. Dually, 
the final coalgebras can be obtained by an inverse limit 
process. In order to explore the limits of this method we look at Heyting 
algebras which have mixed rank 0-1 axiomatizations. We will see that 
Heyting algebras are special in that they are almost rank 1 axiomatized and 
can be handled by a slight variant of the rank 1 coalgebraic methods.
\end{abstract}

\maketitle

\section{Introduction}

Coalgebraic methods and techniques are becoming increasingly important 
in investigating non-classical logics, e.g., \cite{yde:handbook}. In particular, 
logics axiomatized by rank 1 axioms admit coalgebraic representation as 
coalgebras for a functor \cite{AdamekTrnkova90}, \cite{KurzRos07}, 
\cite{PattSchr06}. We recall that an equation is of rank 1 for an operation $f$ 
if each variable occurring in the equation is under the scope of exactly one 
occurrence of $f$. As a result, the algebras for these logics become algebras 
for a functor over the category of underlying algebras without the operation $f$. 
Consequently, free algebras are initial algebras in the category of algebras for 
this functor. This correspondence immediately gives a constructive 
description of free algebras for rank 1 logics relative to algebras in the reduced 
type \cite{Ghilardi95}, \cite{Abramsky05}, \cite{bezh-kurz:calco07}. Examples of 
rank 1 logics 
(relative to Boolean algebras or bounded distributive lattices) are the basic modal 
logic {\bf K}, basic positive modal logic, graded modal logic, probabilistic modal logic, 
coalition logic  and so on; see, e.g., \cite{PattSchr06}. For a coalgebraic approach 
to the complexity of rank 1 logics we refer to \cite{PattSchr06}. On the other hand, 
rank 1 axioms are too simple---very few well-known logics are axiomatized by rank 
1 axioms.  Therefore,  one would want to extend the existing coalgebraic techniques 
to non-rank 1 logics. However, as follows from \cite{KurzRos07}, algebras for these 
logics cannot be represented as algebras for a functor and we cannot use the standard 
construction of free algebras in a straightforward way. 
\\[2ex]
This paper is an extended version of \cite{be-ge:calco09}. However, unlike 
\cite{be-ge:calco09}, here we give a complete solution to the problem of describing 
finitely generated free Heyting algebras in a systematic way using methods similar 
to those used for rank 1 logics. This paper together with \cite{be-ge:calco09} and 
\cite{bezh-kurz:calco07} is a facet of a larger joint project with Alexander Kurz on 
coalgebraic treatment of modal logics beyond rank 1. We recall that an equation is 
of rank 0-1 for an operation $f$ if  each  variable occurring in the equation is under 
the scope of at most one occurrence of $f$. With the ultimate goal of generalizing a 
method of constructing free algebras for varieties axiomatized by rank 1 axioms to 
the case of rank 0-1 axioms, we consider the case of Heyting algebras (intuitionistic 
logic, which is of rank 0-1 for $f=\to$). In particular, we construct free Heyting algebras. 
For an extension of coalgebraic techniques to deal with the finite model property of 
non-rank 1 logics we refer to \cite{SchrPatt08}.
\\[2ex]
Free Heyting algebras have been the subject of intensive investigation for decades. 
The one-generated free Heyting algebra was constructed in \cite{Rieger49} and 
independently in \cite{Nishimura60}. An algebraic characterization of finitely generated 
free Heyting algebras is given in \cite{Urq73}. A very detailed description of finitely 
generated free Heyting algebras in terms of their dual spaces was obtained independently 
in \cite{Shehtman78}, \cite{Grigolia87}, \cite{Bell87} and \cite{Rybakov92}. This method 
is based on a description of the points of finite depth of the dual frame of the free Heyting 
algebra. For a detailed overview of this construction we refer to \cite[Section 8.7]{CZ97} and 
\cite[Section 3.2]{NickThesis}. Finally, \cite{Ghilardi92} introduces a different method for 
describing free Heyting algebras. In \cite{Ghilardi92} the free Heyting algebra is built on a 
distributive lattice step-by-step by freely adding to the original lattice the implications of 
degree $n$, for each $n\in \omega$. \cite{Ghilardi92} uses this technique to show that 
every finitely generated free Heyting algebra is a bi-Heyting algebra. A more detailed 
account of this construction, which we call {\em Ghilardi's construction} or {\em Ghilardi's 
representation},  can be found in \cite{Butz98} and \cite{GZ02}. Based on this method, 
\cite{GZ02} derives a model-theoretic proof of Pitts' uniform interpolation theorem. In 
\cite{AGM07} a similar construction is used to describe free linear Heyting algebras over 
a finite distributive lattice and \cite{dito08} applies the same method to construct
higher order cylindric Heyting algebras. Recently, in \cite{Ghilardi10} this approach was 
extended to {\bf S4}-algebras of modal logic.
We also point out that in \cite{Ghilardi95} a systematic study is undertaken on the connection between 
constructive descriptions of free modal algebras and the theory of normal forms. 
\\[2ex]
Our contribution is to derive Ghilardi's representation of finitely generated free Heyting 
algebras in a modular way which is based entirely on the ideas of the coalgebraic 
approach to rank 1 logics, though it uses these ideas in a non-standard way. We split 
the process into two parts. We first apply the initial algebra construction to weak and 
pre-Heyting algebras---these are consecutive rank 1 approximants of Heyting algebras. 
We then use a non-standard colimit system based on the sequence of algebras for building 
free pre-Heyting algebras in the standard coalgebraic framework.
\\[2ex]
The closest approximants of the rank 1 reducts of Heyting algebras appearing in the 
literature are weak Heyting algebras introduced in  \cite{CJ05}. This is in fact why we 
first treat weak Heyting algebras. However, Heyting algebras satisfy more rank 1 axioms, 
namely at least those of what we call pre-Heyting algebras. While we identified these 
independently in this work, these were also identified by Dito Pataraia, also in connection 
with the study of Ghilardi's construction\footnote{Private communication with the first listed 
author.}. The fact that the functor for which pre-Heyting algebras are algebras yields 
Ghilardi's representation by successive application and quotienting out well-definition of 
implication is, to the best of our knowledge, new to this paper.
\\[2ex]
On the negative side, we use some properties particular to Heyting algebras, and thus 
our work does not yield a method that applies in general. Nevertheless, we expect that 
the approach, though it would have to be tailored, is likely to be successful in other
instances as well. Obtained results allow us to derive a coalgebraic representation for 
weak and pre-Heyting algebras and sheds new light on the very special nature of Heyting 
algebras.
\\[2ex]
The paper is organized as follows. In Section 2 we recall the so-called Birkhoff (discrete) 
duality for distributive lattices. We use this duality in Section 3 to build free weak Heyting 
algebras and in Section 4 to build free pre-Heyting algebras. Obtained results are applied 
in Section 5 for describing free Heyting algebras. In Section 6 we give a coalgebraic 
representation for weak and pre-Heyting algebras. We conclude the paper by listing some 
future work.

\vspace{-3mm}

\section{Discrete duality for distributive lattices}\label{sec:DLduality}

We recall that a {\em bounded distributive lattice} is a distributive lattice with the 
least and greatest elements $0$ and $1$, respectively. Throughout this paper we 
assume that all the distributive lattices are  bounded. We also recall that a non-zero 
element $a$ of a distributive lattice $D$ is called {\it join-irreducible} if for every 
$b,c\in D$ we have that $a\leq b\vee c$ implies $a\leq b$ or $a\leq c$. For each 
distributive lattice (DL for short) $D$ let $(J(D),\leq)$ denote the 
subposet of $D$ consisting of all join-irreducible elements of $D$. Recall also that 
for every poset $X$ a subset $U\subseteq X$ is called a {\it downset} if $x\in U$ and 
$y\leq x$ imply $y\in U$. For each poset $X$ we denote by $\mathcal{O}(X)$ the 
distributive lattice $(\mathcal{O}(X), \cap, \cup,\emptyset, X)$ of all downsets of $X$. 
Finite Birkhoff duality tells us that every finite distributive lattice $D$ is isomorphic to 
the lattice of all downsets of $(J(D),\leq)$ and vice versa, every finite poset $X$ is 
isomorphic to the poset of join-irreducible elements of $\mathcal{O}(X)$. We call 
$(J(D),\leq)$ the {\it dual poset} of $D$ and we call $\mathcal{O}(X)$ the {\it dual lattice} 
of $X$. 

Recall that a lattice morphism is called {\em bounded} if it preserves $0$ and $1$.  
Unless stated otherwise all the lattice morphisms will be assumed to be bounded. 
The duality between finite distributive lattices and finite posets can be extended to 
the duality between the category ${\bf DL}_\mathrm{fin}$ of finite distributive lattices and 
lattice morphisms and the category ${\bf Pos}_\mathrm{fin}$ of finite posets and 
order-preserving maps. In fact, if $h: D \to D'$ is a lattice morphism, then the restriction 
of $h^\flat$, the lower adjoint of $h$,  to $J(D')$ is an order-preserving map between 
$(J(D'), \leq')$ and $(J(D),\leq)$,  and if  $f:X \to X'$ is an order-preserving map between
two posets $X$ and $X'$, then 
$f^{\downarrow}:\mathcal{O}(X)\to \mathcal{O}(X'),\ S\mapsto\ {\downarrow}f(S)$ is 
$\bigvee$-preserving and its upper adjoint 
$(f^{\downarrow})^\sharp=f^{-1}: \mathcal{O}(X') \to \mathcal{O}(X)$ is a lattice morphism. 
Moreover, injective lattice morphisms (i.e.\ embeddings or, equivalently, regular 
monomorphisms) correspond to surjective order-preserving maps, and surjective lattice 
morphisms (homomorphic images) correspond to order embeddings (order-preserving 
and order-reflecting injective maps) that are in one-to-one correspondence with subposets 
of the corresponding poset. 
A particular case of this correspondence that will be used repeatedly throughout this paper 
is spelled out in the following proposition.

\begin{prop}
\label{prop:SubQuot}
Let $D$ be a finite distributive lattice, $X=(J(D),\leq)$ its dual poset, and $a,b\in D$. The 
least-collapsed quotient of $D$ in which $a\leq b$ is the dual lattice of the subposet of $X$ based on the set  
$\{y\in X\mid y\leq a \implies y\leq b\}$.
\end{prop}

\begin{proof}
By Birkhoff duality, every quotient of $D$ corresponds to a subposet of $X$. For any subposet 
$Y\subseteq X$, the dual quotient is given by
\begin{align*}
\mathcal O(X)&\quad\to\quad \mathcal O(Y)\\
         U           &\quad\mapsto\quad U\cap Y.
\end{align*} 
For $a,b\in D$, $a$ corresponds to the downset $\widehat{a}=\{x\in X\mid x\leq a\}$ and 
similarly for $b$ and we have
\[
\widehat{a}\cap Y\subseteq \widehat{b}\cap Y\qquad\iff\qquad \forall y\in Y\quad(y\leq a\ \implies\ y\leq b).
\]
Clearly, the largest subset $Y$ for which this is true is $Y=\{y\in X\mid y\leq a \implies y\leq b\}$.
\end{proof}

We will also need the following fact.

\begin{prop}
\label{prop:dualfreeDL}
Let $X$ be a finite set and $F_{DL}(X)$ the free distributive lattice over $X$. Then the 
poset $(J(F_{DL}(X)), \leq )$ of join-irreducible elements of $F_{DL}(X)$ is isomorphic to 
$(\mathcal{P}(X),\supseteq)$, where $\mathcal{P}(X)$ is the powerset of $X$ and each 
subset $S\subseteq X$ corresponds to the conjunction $\bigwedge S\in F_{DL}(X)$. 
Moreover, for $x\in X$ and $S\subseteq X$ we have 
\begin{center}
$\bigwedge S\leq x$ iff $x\in S$. 
\end{center}
\end{prop}

\begin{proof}
This is equivalent to the disjunctive normal form representation for elements of $F_{DL}(X)$.
\end{proof}

Finally we recall that an element $a\neq 1$ in a distributive lattice $D$ is called 
{\it meet-irreducible} provided, for every $b,c\in D$, we have that $b\wedge c\leq a$ 
implies $b\leq a$ or $c\leq a$. We let $M(D)$ denote the set of all meet-irreducible 
elements of $D$. 

\begin{prop}\label{cor1}
Let $D$ be a finite distributive lattice. Then for each $p\in J(D)$, there exists 
$\kappa(p)\in M(D)$ such that $p\nleq\kappa(p)$ and for each $a\in D$ we have 
$$
p\leq a\ \ \mbox{or}\ \ a\leq \kappa(p).
$$
\end{prop}
\begin{proof}
For $p\in J(D)$, let $\kappa(p)=\bigvee\{a\in D\mid p\nleq a\}$. Then it is clear that the 
condition involving all $a\in D$ holds. Note that if $p\leq\kappa(p)=\bigvee\{a\in D\mid p\nleq a\}$, 
then, applying the join-irreducibility of $p$, we get $a\in D$ with $p\nleq a$ but $p\leq a$, 
which is clearly a contradiction. So it is true that $p\nleq\kappa(p)$. Now we show that 
$\kappa(p)$ is meet irreducible. First note that since $p$ is not below $\kappa(p)$, the 
latter cannot be equal to $1$. Also, if $a,b\nleq\kappa(p)$ then $p\leq a,b$ and therefore 
$p\leq a\wedge b$. Thus it follows that $a\wedge b\nleq\kappa(p)$. This concludes the 
proof of the proposition. 
\end{proof}

\begin{rem}
Note that $p$ and $\kappa(p)$ satisfying the condition of Proposition~\ref{cor1} are called 
a {\em splitting} and {\em co-splitting elements of $D$}, respectively. The notions of splitting 
and co-splitting elements were introduced in lattice theory  by Whitman \cite{Whitman43}. 
It is easy to see that if for $p\in D$ there exists $\kappa'(p)\in D$ satisfying the condition of 
Proposition~\ref{cor1}, then $\kappa'(p) = \kappa(p)$. Therefore, $\kappa(p)\in D$, 
satisfying the condition of Proposition~\ref{cor1}, is unique. 
\end{rem}

\section{Weak Heyting algebras}\label{sec:wHA}

In this section we introduce weak Heyting algebras and describe  
finitely generated free weak Heyting algebras.

\subsection{Freely adding weak implications}

\begin{defi}\label{wHA}\cite{CJ05}
A pair $(A,\to)$ is called a 
{\em  weak  Heyting algebra}\footnote{In \cite{CJ05}  weak Heyting algebras are 
called `weakly Heyting algebras'.}, wHA for short, if $A$ is a distributive lattice
and $\ \to: A^2 \to A$ is a {\em weak implication}, that is, a binary operation satisfying 
the following axioms for all $a,b,c\in A$:
\begin{enumerate}[(1)]
\item  $a\to a=1$,

\item  $a\to (b\wedge c)=(a\to b)\wedge (a\to c)$.

\item  $(a\vee b)\to c = (a\to c) \wedge (b\to c)$.

\item  $(a\to b)\wedge (b\to c) \leq a\to c$.
\end{enumerate}
\end{defi}

\noindent
It is easy to see that by (2) weak implication is order-preserving in the second coordinate and 
by (3) order-reversing in the first. The following lemma gives a few useful properties of wHAs.

\begin{lem}\label{wHA:lem}
Let $(A,\to)$ be a wHA. For each $a,b\in A$ we have 
\begin{conditionsiii}
\item $a\to b =  a\to(a\wedge b)$,
\item $1\to(a\to b)\leq(1\to a)\to(1\to b)$.
\end{conditionsiii}
\end{lem}
\begin{proof}
(i) By axiom (2) we have $a\to(a\wedge b)=(a\to a)\wedge(a\to b)$ and by  
axiom (1) we have $a\to a=1$ so we obtain $a\to b=a\to(a\wedge b)$. 

(ii) Since the weak implication $\to$ is order-reversing in the first coordinate and  
$1\geq 1\to a$ we have $1\to(a\to b)\leq(1\to a)\to(a\to b)$. By (i) of this lemma we obtain  
$(1\to a)\to(a\to b)=(1\to a)\to[(1\to a)\wedge(a\to b)]$. Now axiom (4) yields  
$(1\to a)\wedge(a\to b)\leq 1\to b$.  So $(1\to a)\to(a\to b)\leq(1\to a)\to(1\to b)$. By 
transitivity of the order we have the desired result. 
\end{proof}            
          
\vspace{2mm}

\noindent
Let $D$ and $D'$ be distributive lattices. We let $\nto\hspace{-1mm}(D\times D)$ denote 
the set $\{a\nto b: a\in D$ and $b\in D'\}$. We stress that this is just a set bijective with 
$D\times D'$. The symbol $\nto$ is just a formal notation. For each distributive lattice $D$ 
we let $F_{DL}(\nto\hspace{-1mm}(D\times D))$ denote the free distributive lattice over 
$\nto\hspace{-1mm}(D\times D)$. Moreover, we let 
\[
\qquad\qquad H(D)=F_{DL}(\nto\hspace{-1mm}(D\times D))/_\approx
\] 
\noindent where $\approx$ is the DL congruence generated by the axioms (1)--(4) seen 
as relation schemas for $\nto$. The point of view is that of describing a distributive lattice 
by generators and relations. That is, we want to find the quotient of the free  distributive lattice
over the set  $\nto\hspace{-0.5mm}(D\times D)$ with respect to the lattice congruence generated
by the pairs of elements of $F_{DL}(\nto\hspace{-0.5mm}(D\times D))$ in (1)--(4) (where $\to$ is 
replaced by $\nto$) with $a,b,c$ ranging over $D$. For an element 
$a\nto b\in F_{DL}(\nto\hspace{-0.5mm}(D\times D))$ we denote by $[a\nto b]_\approx$ 
the $\approx$ equivalence class of $a\nto b$.

The rest of this subsection will be devoted to showing that for each finite distributive lattice $D$ the 
poset $(J(H(D)), \leq)$ is isomorphic to $(\mathcal{P}(J(D)), \subseteq)$. Below we give a dual 
proof of this fact. The dual proof,  which relies on the fact that identifying two elements of an 
algebra simply corresponds to throwing out those points of the dual that are below one and not 
the other (Proposition~\ref{prop:SubQuot}), is produced in a modular and systematic way that does 
not require any prior insight into the structure of these particular algebras.
\\[2ex]
We start with a finite distributive lattice $D$ and 
the free DL generated by the set
\[
\nto\hspace{-0.5mm}(D\times D)=\{a\nto b \mid a,b\in D\}
\]  
 \noindent
of all formal arrows over $D$. As follows from Proposition~\ref{prop:dualfreeDL}, 
$J(F_{DL}(\nto\hspace{-0.5mm}(D\times D)))$ is isomorphic to the powerset of 
$\nto\hspace{-0.5mm}(D\times D)$, ordered by reverse inclusion. Each subset  
$S \subseteq \nto\hspace{-0.5mm}(D\times D)$ corresponds to the conjunction 
$\bigwedge S$ of the elements of $S$; the empty set of course corresponds to $1$.
Now we want to take quotients of this free distributive lattice with respect to various 
lattice congruences, namely the ones generated by the set of instances of the axioms 
for weak Heyting algebras.
\\[2ex]
\noindent{\bf The relational schema $x\nto x \approx 1$}.
\\[1ex]
Here we want to take the quotient of $\Fdd$ with respect to the lattice congruence of $\Fdd$ 
generated by the set $\{(a\nto a,1)\mid a\in D\}$. By Proposition~\ref{prop:SubQuot} this quotient is given dually by the  
subposet, call it $P_1$, of our initial poset $P_0=J(\Fdd)$, consisting of those join-irreducibles of 
$\Fdd$ that do not violate this axiom. Thus, for $S\in J(\Fdd)$, $S$ is admissible provided
\[
\forall a\in D \qquad (\bigwedge S\leq 1 \quad\iff\quad \bigwedge S\leq a\nto a).
\]  
Since all join-irreducibles are less than or equal to $1$, it follows that the only join-irreducibles 
that are admissible are the ones that are below $a\nto a$ for all $a\in D$. That is,viewed as 
subsets of $\nto\hspace{-0.5mm}(D\times D)$, only the ones that contain $a\nto a$ for each $a\in D$:
\[
P_1=\{S\in P_0\mid a\nto a\in S\mbox{ for each }a\in D\}. 
\] 

\noindent{\bf The relational schema $x\nto(y\wedge z) \approx (x\nto y)\wedge(x\nto z)$}.
\\[1ex]
We now want to take a further quotient and thus we want to keep only those join-irreducibles from 
$P_1$ that do not violate this relational schema. That is, $S\in P_1$ is admissible provided
\[
\forall a,b,c\qquad (\bigwedge S \leq a\nto (b\wedge c) 
                      \quad\iff\quad \bigwedge S \leq a\nto b\ \mbox{ and }\ \bigwedge S \leq a\nto c).
\] 
\noindent
which means 
\[
\forall a,b,c\qquad (a\nto(b\wedge c)\in S 
                      \quad\iff\quad a\nto b\in S\ \mbox{ and }\ a\nto c\in S).
\] 

Let $P_2$ denote the poset of admissible join-irreducible elements of $P_1$.  

\begin{prop}
The poset $P_2$ is order isomorphic to the set
\[
Q_2=\{f:D\to D\mid \forall a\in D,  f(a)\leq a\}
\]
ordered pointwise. 
\end{prop}

\begin{proof}
An admissible $S$ from $P_2$ corresponds to the function $f_S:D\to D$ given by
\[
f_S(a)=\bigwedge\{b\in D\mid a\nto b\in S\}.
\]
In the reverse direction, a function in $Q_2$ corresponds to the admissible set 
\[
S_f=\{a\nto b\mid f(a)\leq b\}.
\]
The proof that this establishes an order isomorphism is a straightforward verification.
\end{proof}
\medskip
\noindent{\bf The relational schema $(x\vee y)\nto z \approx (x\nto z)\wedge(y\nto z)$}.
\\[1ex]

We want the subposet of $Q_2$ consisting of those $f$'s such that
\[
\forall a,b,c\qquad \big((a\vee b)\nto c\in S_f 
                  \quad\iff\quad a\nto c\in S_f\ \mbox{ and }\ b\nto c\in S_f \big).
\] 
To this end notice that 
\begin{align*}
&\forall a,b,c\qquad \big((a\vee b)\nto c\in S_f 
                  \ \iff\ (a\nto c\in S_f\ \mbox{ and }\ b\nto c\in S_f)\big)\\
\iff\quad &
\forall a,b,c\qquad \big(f(a\vee b)\leq c 
                   \quad\iff\quad (f(a)\leq c\ \mbox{ and }\ f(b)\leq c)\big)\\
\iff\quad &
\forall a,b\qquad\quad  f(a\vee b)=f(a)\vee f(b).
\end{align*}
That is, the poset $P_3$ of admissible join-irreducibles left at this stage is isomorphic to  the set
\[
Q_3=\{f:D\to D\mid \ f\ \mbox{ is join-preserving and }\ \forall a\in D\quad f(a)\leq a\}.
\]

\noindent{\bf The relational schema $(x\nto y)\wedge(y\nto z) \precapprox x\nto z$}.
\\[1ex]
It is not hard to see that this yields, in terms of
join-preserving
functions  $f:D\to D$,
\begin{align*}
Q_4 &=\{f\in Q_3\mid\forall a\in D\ f(a)\leq f(f(a))\}\\
         &=\{f:D\to D\mid f\mbox{ is join-preserving and }\forall a\in D\ f(a)\leq f(f(a))\leq f(a)\leq a\}\\
         &=\{f:D\to D\mid f\mbox{ is join-preserving and }\forall a\in D\ f(f(a))= f(a)\leq a\}.
\end{align*}
We note that the elements of $Q_4$ are nuclei \cite{Joh82} on the order-dual lattice of $D$.
Since the $f$'s in $Q_4$ are join and $0$  preserving, they are completely given by their 
action on $J(D)$. The additional property shows that these functions have lots
of fixpoints. In fact, we can show that they are completely described by their join-irreducible 
fixpoints.

\begin{lem}\label{lemma3.4}
Let $f\in Q_4$, then for each $a\in D$ we have 
\[
f(a)=\bigvee\{r\in J(D)\mid f(r)=r\leq a\}.
\]
\end{lem}
\begin{proof}
Clearly $\bigvee\{r\in J(D)\mid f(r)=r\leq a\}\leq f(a)$. For the converse, let $r$ be
maximal in $J(D)$ with respect to the property that $r\leq f(a)$. Now it follows that
\[
r\leq f(a)=f(f(a))=\bigvee\{f(q)\mid J(D)\ni q\leq f(a)\}.
\]
Since $r$ is join-irreducible, there is $q\in J(D)$ with $q\leq f(a)$ and $r\leq f(q)$.
Thus $r\leq f(q)\leq q\leq f(a)$ and by maximality of $r$ we conclude that $q=r$.
Now $r\leq f(q)$ and $q=r$ yields $r\leq f(r)$. However, $f(r)\leq r$ as this holds for 
any element of $D$ and thus $f(r)=r$. Since any element in a finite lattice is the join 
of the maximal join-irreducibles below it, we obtain
\begin{align*}
f(a) & =\bigvee\{r\in J(D)\mid r\mbox{ is maximal in $J(D)$ with respect to }r\leq f(a)\}\\
       & \leq \bigvee\{r\in J(D)\mid f(r)=r\leq f(a)\} \ \leq \ f(a).
\end{align*}
Finally, notice that if $f(r)=r\leq f(a)$ then as $f(a)\leq a$, we have 
$f(r)=r\leq a$. Conversely, if  $f(r)=r\leq a$ then $r=f(r)=f(f(r))\leq f(a)$ and we have
proved the lemma.
\end{proof}

\begin{prop}\label{prop:jhd}
The set of functions in $Q_4$, ordered pointwise, is order isomorphic to 
the powerset of $J(D)$ in the usual inclusion order.
\end{prop}

\begin{proof}
The order isomorphism is given by the following one-to-one correspondence 
\begin{align*}
Q_4\ \  & \leftrightarrows \ \  {\mathcal P}(J(D))\\
  f   \ \    & \mapsto     \ \   \{p\in J(D)\mid f(p)=p\}\\
 f_T\ \   & \mapsfrom \ \  \ T  
\end{align*}
where $f_T:D\to D$ is given by $f_T(a)=\bigvee\{p\in J(D)\mid T\ni p\leq a\}$.
Using Lemma~\ref{lemma3.4}, it is straightforward to see that these two assignments are 
inverse to each other. Checking that $f_T$ is join preserving and satisfies $f^2=f\leq id_D$
is also straightforward. Finally, it is clear that $f_T\leq f_S$ if and only if $T\subseteq S$.
\end{proof}

\noindent
Next we will prove a useful lemma that will be applied often throughout the 
remainder of this paper. 
Let $D$ be a finite distributive lattice. For $a,b\in D$ and $T\subseteq J(D)$ we write 
$T\preceq a\nto b$ provided $\bigwedge S_T \leq a\nto b$, where 
$S_T = S_{f_{T}}= \{a\nto b: f_T(a)\leq b\}$.

\begin{lem}\label{imp:lem}
Let $D$ be a finite distributive lattice. For each $a,b\in D$ and $T\subseteq J(D)$ we  
have
\begin{center}
$T\preceq  a\nto b$ \ \ iff \ \ $\forall p\in T \ (p\leq a \mbox{ implies } p\leq b)$
\end{center}
\end{lem}
\begin{proof}
\begin{align*}
T\preceq a\nto b &\iff  \bigwedge S_T \leq a\nto b\\
                   &\iff  a\nto b\in S_T\\
                   &\iff f_T(a)\leq b\\
                   &\iff \bigvee ({\downarrow}a\cap T)\leq b\\
                  &\iff \forall p\in T\ (p\leq a \mbox{ implies } p\leq b).
                \end{align*}
\end{proof}

This subsection culminates in the following theorem. Recall that, for a join-irredu\-cible $q$ 
of a finite distributive lattice, $\kappa(q)$ is the corresponding meet irreducible, recall 
Proposition~\ref{cor1}.

\begin{thm}
\label{thrm:JHD}
Let $D$ be a finite distributive lattice and $X=(J(D), \leq)$ its dual poset. Then the following
statements are true:
\begin{conditionsiii}
\item The poset $(J(H(D)), \leq)$ is isomorphic to the poset $(\mathcal{P}(X),\subseteq)$ 
of all subsets of $X$ ordered by inclusion.
\item $J(H(D)) = \{[\bigwedge_{q\not\in T}(q \nto \kappa(q))]_{\approx}\mid   T\subseteq X\}$.
\end{conditionsiii}
\end{thm}
\begin{proof}
As shown above, the poset $J(H(D))$, obtained from $J(F_{DL}({\nto}(D\times D)))$ by 
removing the elements that violate the relational schemas obtained from axioms (1)--(4),  is 
isomorphic to the poset $Q_4$, and $Q_4$ is in turn isomorphic to $\mathcal{P}(J(D))$ ordered 
by inclusion, see Proposition~\ref{prop:jhd}.

In order to prove the second statement, let $q\in J(D)$, and consider 
$q\nto\kappa(q)\in F_{DL}(\nto(D\times D))$. If we represent $H(D)$ as the lattice of downsets 
${\mathcal O}(J(H(D)))$, then the action of the quotient map on this element is given by
\begin{align*}
\qquad F_{DL}(\nto(D\times D))\ &\to\  H(D)\\
   q\nto\kappa(q)    \qquad         \ &\mapsto\ \{T'\in{\mathcal P}(J(D))\mid q\nto\kappa(q)\in S_{T'}\}. 
\end{align*}
Now 
\vskip-.7cm
\begin{align*}
\qquad q\nto\kappa(q)\in S_{T'} & \iff f_{T'}(q)\leq \kappa(q)\\
				   & \iff \bigvee({\downarrow}q\cap T')\leq \kappa(q)\\
				   & \iff q\not\in T'.
\end{align*}
The last equivalence follows from the fact that $a\leq\kappa(q)$ if and only if $q\nleq a$ 
and the only element of ${\downarrow}q$ that violates this is $q$ itself (Proposition~\ref{cor1}). 
We now can see that for any $T\subseteq J(D)$ we have 
\begin{align*}
\Fdd \quad &\to\quad  H(D)\\
[\bigwedge_{q\not\in T} (q \nto \kappa(q))]_{\approx} \quad 
	&\mapsto\quad \{T'\in{\mathcal P}(J(D))\mid \forall q\quad(q\not\in T\ \Rightarrow\  q\nto\kappa(q)\in S_{T'})\}\\
	&\ \quad =\{T'\in{\mathcal P}(J(D))\mid \forall q\quad(q\not\in T\ \Rightarrow\ q\not\in T')\}\\
	&\ \quad =\{T'\in{\mathcal P}(J(D))\mid \forall q\quad(q\in T'\ \Rightarrow\ q\in T)\}\\
	&\ \quad =\{T'\in{\mathcal P}(J(D))\mid T'\subseteq T\}.
\end{align*}
That is, under the quotient map $\Fdd\ \to\  H(D)$, the elements 
$\bigwedge_{q\not\in T} (q \nto \kappa(q))$ are mapped to the principal downsets ${\downarrow}T$,
for each $T\in{\mathcal P}(J(D))=J(H(D))$. Since these principal downsets are exactly the 
join-irreducibles of ${\mathcal O}(J(H(D)))=H(D)$, we have that 
$\{\ [\bigwedge_{q\not\in T}(q \nto \kappa(q))]_{\approx}\ \mid\   T\subseteq J(D)\ \}=J(H(D))$.
\end{proof}

\noindent
It follows from Theorem~\ref{thrm:JHD}(1) that if two finite distributive lattices $D$ and $D'$ have 
an equal number of join-irreducible elements, then $H(D)$ is isomorphic to $H(D')$. To see this, 
we note that if $|J(D)|=|J(D')|$, then $({\mathcal P}(J(D)), \subseteq)$ is isomorphic to 
$({\mathcal P}(J(D')), \subseteq)$. This,  by Theorem~\ref{thrm:JHD}(1), implies that $H(D)$ is 
isomorphic to $H(D')$. In particular, any two non-equivalent orders on any finite set give rise to 
two non-isomorphic distributive lattices with isomorphic $H$-images.

\begin{rem}
All the results in this section for finite distributive lattices can be generalized to the infinite case. 
In the infinite case, however, instead of  finite posets we would need to work with Priestley spaces 
and instead of the finite powerset we need to work with the Vietoris space (see Section~\ref{top:sec}).
As we will see in Sections~\ref{subsec:FreewHA}, \ref{sec:PreHA} and \ref{sec:HA}, 
for our purposes (that is, for describing finitely generated free weak Heyting, pre-Heyting and Heyting 
algebras), it suffices to work with limits of finite distributive lattices. So we will stick with the finite 
case, for now, and will consider infinite distributive lattices and Priestley spaces only in Section 6, 
where we discuss coalgebraic representation of weak Heyting and pre-Heyting algebras.
\end{rem}

\subsection{Free weak Heyting algebras}\label{subsec:FreewHA}

In the coalgebraic approach to generating the free algebra, it is a fact of central importance that $H$ 
as described here is actually a functor. That is, for a DL homomorphism $h:D\to E$ one can define a 
DL homomorphism $H(h):H(D)\to H(E)$ so that $H$ becomes a functor on the category of DLs.
To see this, we only need to note that $H$ is defined by rank 1 axioms. We recall that for an operator 
$f$ (in our case $f$ is the weak implication $\to$) an equation is of {\em rank 1}  if each variable in the 
equation is under the scope of exactly one occurrence of $f$ and an equation is of {\em rank 0-1}  
if each variable in the equation is under the scope of at most one occurrence of $f$. It is easy to check
that  axioms (1)--(4) for weak Heyting algebras are rank 1. Therefore, $H$ gives rise to a functor 
$H:{\bf DL} \to {\bf DL}$, where ${\bf DL}$ is the category of all distributive lattices and lattice morphisms. 
(This fact can be found in \cite{AdamekTrnkova90}, for $\Set$-functors, and in \cite{KurzRos07} for the 
general case.) Moreover, the category of weak Heyting algebras is isomorphic to the category $Alg(H)$ 
of the algebras for the functor $H$. For the details of such correspondences we refer to either of 
\cite{Abramsky05}, \cite{AdamekTrnkova90}, \cite{bezh-kurz:calco07}, \cite{Ghilardi95}, \cite{KurzRos07}.
We would like to give a concrete description of how $H$ applies to DL homomorphisms. We describe 
this in algebraic terms here and give the dual construction via Birkhoff duality.

Let $h:D\to E$ be a DL homomorphism. Recall that the dual map from $J(E)$ to $J(D)$ is just the 
lower adjoint $h^\flat$ with domain and codomain properly restricted. By abuse of notation we will 
just denote this map by $h^\flat$, leaving it to the reader to decide what the proper domain and 
codomain is. Now $H(D)=F_{DL}(\nto(D\times D))/\approx_D$, where $\approx_D$ is the DL congruence 
generated by the set of all instances of the axioms (1)--(4) with $a,b,c\in D$. Also let $q_D$ be the 
quotient map corresponding to quotienting out by $\approx_D$. Any lattice homomorphism 
$$h:D\to E$$ 
yields a map 
$$h\times h: D\times D\longrightarrow E\times E$$ 
and this of course yields a lattice homomorphism 
$$F_{DL}(h\times h): F_{DL}(\nto(D\times D))\longrightarrow F_{DL}(\nto(E\times E)).$$
Now the point is that  $F_{DL}(h\times h)$ carries elements of $\approx_D$ to elements of $\approx_E$
(it is an easy verification and only requires $h$ to be a homomorphism for axiom schemas (2) and (3)). 
Thus we have $\approx_D\subseteq Ker(q_E\circ F_{DL}(h\times h))$ or equivalently that there is a 
unique map $H(h):H(D)\to H(E)$ that makes the following diagram commute
\vspace*{-.2cm}
\begin{center}
$\xymatrix@M=5pt{
F_{DL}(\to(D\times D))\ \ar[rr]^{F_{DL}(h\times h)} \  \ar@{->>}[d]^{q_D} & &F_{DL}(\to(E\times E))
\ar@{->>}[d]^{q_E} \\ 
H(D) \ \ar@{-->}[rr]^{H(h)}& &H(E). }$ 
\end{center}
The dual diagram is  
\begin{center} 
$\xymatrix@M=5pt{
{\mathcal P}(D\times D)\ \ar@{<-}[rr]^{\ \ (h\times h)^{-1}} \   & &{\mathcal P}(E\times E) \\
{\mathcal P} (J(D)) \ \ar@{_(->}[u]^{e_D}\ \ar@{<--}[rr]^{{\mathcal P}(h^\flat)}& &{\mathcal P}(J(E)) 
\ar@{_(->}[u]^{e_E}} $
\end{center} 
The map $e_D:{\mathcal P} (D)\hookrightarrow  {\mathcal P}(D\times D)$ is the embedding, via 
$Q_4$ and so on into $P_0$ as obtained above. That is, 
$e_D(T)=\{a\to b\mid \forall p\in T\ (p\leq a\Rightarrow p\leq b)\}$. Now in this dual setting, the fact 
that there is a map ${\mathcal P}(h^\flat)$ is equivalent to the fact that $(h\times h)^{-1}\circ e_E$ 
maps into the image of the embedding $e_D$. This is easily verified:
\begin{align*} 
(h\times h)^{-1}( e_E(T)) & =\{a\to b\mid \forall q\in T\ (q\leq h(a)\Rightarrow q\leq h(b))\}\\
                                             & =\{a\to b\mid \forall q\in T\ (h^\flat(q)\leq a\Rightarrow h^\flat(q)\leq b)\}\\
                                             & =\{a\to b\mid \forall p\in h^\flat(T)\ (p\leq a\Rightarrow p\leq b)\}\\
                                             &=e_D(h^\flat(T)).
\end{align*}
Thus we can read off directly what the map ${\mathcal P}(h^\flat)$ is: it is just forward image under 
$h^\flat$. That is, if we call the dual of $h:D\to E$ by the name $f:J(E)\to J(D)$, then ${\mathcal P}(f)=f[\ ]$
where $f[\ ]$ is the lifted forward image mapping subsets of $J(E)$ to subsets of $J(D)$. Finally, we 
note that ${\mathcal P}(f)$ is an embedding if and only if $f$ is injective, and ${\mathcal P}(f)$ is surjective 
if and only if $f$ is surjective.

\begin{rem}
It follows from Theorem~\ref{thrm:JHD}(i) that the functor $H$ can be represented as a 
composition of two functors. Let $B: {\bf DL}_\mathrm{fin}\to {\bf BA}_\mathrm{fin}$ be the 
functor from the category of finite distributive lattices to the category of finite Boolean algebras 
which maps every finite distributive lattice to its free Boolean extension---the (unique) Boolean 
algebra generated by this distributive lattice. It is well known \cite{Nerode59} that the dual of 
the functor $B$ is the forgetful functor from the category of finite posets to the category of finite 
sets, which maps every finite poset to its underlying set. Further, let also 
$H_B: {\bf BA}_\mathrm{fin} \to {\bf DL}_\mathrm{fin}$ be the functor $H$ restricted to Boolean 
algebras. That is, given a Boolean algebra $A$ we define $H_B(A)$ as the free DL over 
$\nto\hspace{-1.5mm}(B,B)$ quotiented out by the relational schemas corresponding to the axioms 
(1)--(4) of wHAs. Then the functor which is dual to $H_B$ maps each finite set $X$ to 
$({\mathcal P}(X), \subseteq)$ and therefore $H: {\bf DL}_\mathrm{fin}\to {\bf DL}_\mathrm{fin}$ is 
the composition of $B$ with $H_B$. \end{rem}

Since weak Heyting algebras are the algebras for the functor $H$, we can make use of coalgebraic 
methods for constructing free weak Heyting algebras. Similarly to \cite{bezh-kurz:calco07}, where 
free modal algebras and free distributive modal algebras were constructed, we construct  finitely 
generated free weak Heyting algebras as initial algebras of $Alg(H)$. That is, we have a sequence of 
distributive lattices, each embedded in the next:
\begin{align*}
n\quad& \longrightarrow \ \ F_{DL}(n), \mbox{ the free  distributive lattice on $n$ generators},\\
D_0 \quad& = \quad  F_{DL}(n),\\
D_{k+1} & = \quad D_0 + H(D_k), \mbox{ where $+$ is the coproduct in ${\bf DL}$},\\
i_0\ :\ \ &D_0  \to \ D_0+H(D_0)\ =D_1 \mbox{ the embedding given by coproduct},\\ 
i_k\ :\ \ &D_k \to D_{k+1}  \mbox{ where } i_k=id_{D_0} +H(i_{k-1}).
\end{align*}
For $a,b\in D_k$, we denote by $a\to_k b$ the equivalence class $[a\nto b]_\approx\in H(D_k)
\subseteq D_{k+1}$. 
Now, by applying the technique of  \cite{Abramsky05}, \cite{AdamekTrnkova90}, \cite{bezh-kurz:calco07} 
\cite{Ghilardi95},  to weak Heyting algebras, we arrive at the following theorem.

\begin{thm}\label{freewHA}
The direct limit $(D_\omega, (D_k\to D_\omega)_k)$ in ${\bf DL}$ of the system $(D_k, i_k:D_k\to D_{k+1})_k$ 
with the binary operation $\to_\omega:D_\omega\times D_\omega \to D_\omega$ defined by  
$a \to_\omega b = a \to_k b$, for $a,b\in D_k$ is the free $n$-generated weak Heyting algebra  when we embed 
$n$ in $D_\omega$ via $n\to D_0\to D_\omega$.
\end{thm}

\noindent
Now we will look at the dual of $(D_\omega, \to_\omega)$.  Let $X_0={\mathcal P}(n)$ be the dual of $D_0$ and let 
$$
X_{k+1} = X_0 \times \cl{P}(X_k)
$$
be the dual of $D_{k+1}$.

\begin{thm}\label{freewHAdually} The sequence $(X_k)_{k<\omega}$ with maps $\pi_k:X_0
  \times \cl{P}(X_k)\to X_k$ defined by
$$
\pi_k = id_{X_{0}}\times \mathcal{P}(\pi_{k-1})\ \ \mbox{i.e.}\ \ \pi_k(x,A)=(x,\pi_{k-1}[A])
$$
is dual to the sequence $(D_k)_{k<\omega}$ with maps $i_k: D_k \to
D_{k+1}$. In particular, the $\pi_k$'s are surjective. 
\end{thm}

\begin{proof}
The dual of $D_0$ is $X_0={\mathcal P}(n)$, and since $D_{k+1} =  D_0 + H(D_k)$, 
it follows that $X_{k+1} = X_0 \times \cl{P}(X_k)$ as sums go to products and as $H$ is
dual to $\mathcal P$. For the maps, $\pi_0: X_0\times {\mathcal P}(X_0) \to X_0$ is 
just the projection onto the first coordinate since $i_0$ is the injection given by the 
sum construction. We note that $\pi_0$ is surjective. Now the dual
$\pi_k:X_{k+1}=X_0\times {\mathcal P}(X_k)\to X_k=X_0\times {\mathcal P}(X_{k-1})$ of 
$i_k=id_{D_0}+H(i_{k-1})$ is $id_{X_0}\times {\mathcal P}(\pi_{k-1})$, which is exactly 
the map given in the statement of the theorem. Note that a map of the form $X\times Y\to
X\times Z$ given by $(x,y)\mapsto (x,f(y))$, where $f:Y\to Z$ is surjective if and only the 
map $f$ is. Also, as we saw above,  ${\mathcal P}(\pi_{k})$ is surjective if and only if 
$\pi_k$ is. Thus by induction, all the $\pi_k$'s are surjective.
\end{proof}

\section{Pre-Heyting algebras}\label{sec:PreHA}

In this section we define pre-Heyting algebras which form a subvariety of weak 
Heyting algebras and describe free pre-Heyting algebras. We first note that, for any 
weak Heyting algebra $A$, the map from $A$ to $A$ given by $a\ \mapsto\ (1\to a)$ is 
meet-preserving and also preserves $1$ by virtue of the first two axioms of weak 
Heyting algebras. For the same reason, the map from a distributive lattice $D$ to 
$H(D)$ mapping each element $a$ of $D$ to $[1\nto a]_\approx$ also is meet-preserving 
and preserves $1$. For Heyting algebras more is true: for a Heyting algebra $B$, the map 
given by $b\ \mapsto\ (1\to b)$ is just the identity map and thus, in particular, it is a lattice 
homomorphism. In other words, Heyting algebras satisfy additional rank 1 axioms beyond 
those of weak Heyting algebras.

\begin{defi}\label{QHA}
A weak Heyting algebra $(A,\to)$ is called a {\em pre-Heyting algebra}, {\em pHA} for 
short, if the following additional axioms are satisfied for all $a,b\in A$:

\begin{enumerate}[(5)]
\item  $1\to 0=0$,
\item[(6)]  $(1\to a)\vee(1\to b)=1\to(a\vee b)$.
\end{enumerate}
\end{defi}

Since these are again rank 1 axioms, we can obtain a description of the free finitely generated
pre-Heyting algebras using the same method as for weak Heyting algebras. Accordingly, 
for a finite distributive lattice $D$, similarly to what we did in the previous section, we let 
\[
\qquad\qquad K(D)=\Fdd/_\approxeq
\] 
\noindent where $\approxeq$ is the DL congruence generated by the axioms (1)--(6) viewed as 
relational schemas for $\Fdd$. This of course means we can just proceed from where we left off in  
Section~\ref{sec:wHA} and identify the further quotient of $H(D)$ obtained by the schema 
$(1\nto a)\vee(1\nto b)\approxeq 1\nto(a\vee b)$ for $a$ and $b$ ranging over the elements 
of $D$ and $1\nto 0 \approxeq 0$. That is, we need to calculate
\[
\qquad\qquad K(D)=H(D)/_\simeq
\] 
\noindent  where $\simeq$ is the DL congruence given by the relational schemas obtained from 
the axioms (5)--(6).

We say that a subset $S$ of a poset $(X,\leq)$ is {\em rooted} if there exists $p\in S$, which we call the 
{\em root} of $S$, such that $q\leq p$ for each $q\in S$. It follows from the definition that a root of a 
rooted set is unique.  Also a rooted set is necessarily non-empty.  We denote by ${\mathcal P}_r(X)$ 
the set of all rooted subsets of $(X, \leq)$. We also let 
$$
\mathit{root}: {\mathcal P}_r(X) \to X
$$
be the map sending each rooted subset $S$ of $X$ to its root. It is easy to see that $\mathit{root}$ is surjective 
and order-preserving.

\begin{thm}
\label{thrm:JKD}
Let $D$ be a finite distributive lattice and $X$ its dual poset. Then the following statements are true:
\begin{conditionsiii}
\item The poset $(J(K(D)), \leq)$ is isomorphic to the poset $(\mathcal{P}_r(X),\subseteq)$ of all rooted 
subsets of $X$ ordered by inclusion.
\item $J(K(D)) = \{[(1\nto x)\wedge(\bigwedge_{q<x,q\not\in T'}(q \nto \kappa(q))]_{\approxeq}\mid   
T'\subseteq{\downarrow}x\setminus\{x\}, x\in X\}$.
\item The map $D\rightarrow K(D)$ given by $a\mapsto[1\nto a]_\approxeq$ is an injective lattice 
homomorphism whose dual is the surjective order-preserving map $\mathit{root}:\mathcal{P}_r(X)\to X$.
\end{conditionsiii}
\end{thm}
\begin{proof}
(i) By Theorem~\ref{thrm:JHD}(i), $(J(H(D)), \leq)$ is isomorphic to $(\mathcal{P}(X),\subseteq)$. Thus, we 
need to show that the rooted subsets of $X$ are exactly the subsets which are admissible with respect to 
relational schemas given by the axioms (5) and (6). For the axiom (5), it may be worth clarifying the meaning 
of this relational schema: the $0$ (and $1$) on the left side are elements of $D$, and the expression $1\nto 0$ 
is one of the generators of $\Fdd$, whereas the $0$ on the right of the equality is the bottom of the lattice 
$H(D)$ --- we will denote it by $0_{H(D)}$ for now. A  set $S\subseteq X$ is admissible for (5) provided
\[
S\preceq 1\nto 0\ \iff\ S\preceq 0_{H(D)}.
\]
Now by Lemma~\ref{imp:lem}, $S\preceq 1\nto 0$ if and only if, for all $p\in S$, $p\leq 1$ implies $p\leq 0$. 
Since the former is true for every $p\in X$ and the latter is false for all $p\in X$, the only $S\in\mathcal{P}(X)$ 
satisfying this condition is $S=\emptyset$. On the other hand, as in any lattice, no join-irreducible in $K(D)$ is 
below $0_{K(D)}$. Thus (5) eliminates $S=\emptyset$.

Weak Heyting implication is meet-preserving and thus order-preserving in the second coordinate so that we 
have that $1 \nto (a\vee b) \geq  (1\nto a) \vee (1\nto b)$ already in $H(D)$ for every $D$. Therefore, a set 
$S\subseteq X$ is admissible with respect to (5) and (6) iff $S\neq \emptyset$ and  
\[
S\preceq  1 \nto (a\vee b)\ \mbox{  implies }\ S\preceq 1\nto a\mbox{ or }S\preceq 1\nto b.
\] 
Now by Lemma~\ref{imp:lem}, $S\preceq 1 \nto x$ iff $S\subseteq {\downarrow}x$, so for non-empty $S$ we 
need $S\subseteq {\downarrow}(a\vee b)$ to imply that $S\subseteq {\downarrow}a$ or $S\subseteq {\downarrow}b$ 
for all $a,b\in D$. This is easily seen to be equivalent to rootedness: If $S\subseteq X$ is rooted and $p$ is its root, 
then $S\subseteq {\downarrow}(a\vee b)$ implies $p\leq a\vee b$. Thus, $p\leq a$ or $p\leq b$ and so 
$S\subseteq {\downarrow}a$ or  $S\subseteq {\downarrow}b$. Conversely, if $S$ is admissible then 
$S\neq\emptyset$ and, as it is finite, every element of $S$ is below a maximal element of $S$. If $p\in S$ is 
maximal but not the maximum of $S$ then $S\subseteq  {\downarrow}(p\vee a)$, where $a= \bigvee(S \setminus \{p\})$,
but $S\not\subseteq  {\downarrow}p$ and  $S\not\subseteq {\downarrow} a$.

(ii) The proof is similar to the proof of Theorem~\ref{thrm:JHD}(ii). Recall that for any $T\subseteq X$ we have 
that the join-irreducible $\{T'\in{\mathcal P}(X)\mid T'\subseteq T\}={\downarrow} T$ in ${\mathcal O}({\mathcal P}(X))$,
 which is isomorphic to $H(D)$, is equal to $[\bigwedge_{q\not\in T} (q \nto \kappa(q))]_{\approx}$. Also 
 $[1\nto x]_\approx$ is join-irreducible and corresponds to ${\downarrow} x\subseteq X$. That is,
\[
[1\nto x]_\approx=[\bigwedge_{q\nleq x} (q \nto \kappa(q))]_{\approx}.
\]
Thus, in particular, for $T\subseteq X$ rooted with root $x$ and $T'=T\setminus\{x\}$, the join-irreducible corresponding 
to $T$ is given by
\begin{align*}
[\bigwedge_{q\not\in T} (q \nto \kappa(q)]_{\approx} \quad
&= [\bigwedge_{q\nleq x} (q \nto \kappa(q))\wedge \bigwedge_{q\leq x,q\not\in T} (q \nto \kappa(q))]_{\approx}\\
&=[(1\nto x)\wedge \bigwedge_{q<x,q\not\in T'} (q \nto \kappa(q))]_{\approx}.
\end{align*}
Since this is the case in $H(D)$, it is certainly also true in the further quotient $K(D)$ and in (i) we have shown that all 
join-irreducibles of $K(D)$ correspond to rooted subsets of  $X$, thus the statement follows.

(iii) The map $D\to K(D)$ given by $a\mapsto[1\nto a]_\approxeq$ is clearly a homomorphism since we have quotiented 
out by all the necessary relations: $[1\nto 1]_{\approxeq}=1_{K(D)}$ by (1), $[1\nto 0]_{\approxeq}=0_{K(D)}$ by (5), and 
the map is meet and join preserving by (2) and (6), respectively. As we saw in Section~\ref{sec:DLduality}, the dual of a 
homomorphism between finite lattices is the restriction to join-irreducibles of its lower adjoint, that is, our homomorphism
is dual to the map $r:\mathcal{P}_r(X)\to X$ given by
\[
\forall T\in \mathcal{P}_r(X) \ \forall a\in D \quad (r(T)\leq a\quad\iff\quad T\preceq(1\nto a)).
\]
For $T\in \mathcal{P}_r(X)$ we have $T\preceq (1\nto a)=\bigvee_{x\in X, x\leq a}(1\nto x)$ if and only if there is an 
$x\in X$ with $x\leq a$ and $T\preceq (1\nto x)$. Furthermore $T\preceq (1\nto x)$ if and only if 
$T\subseteq{\downarrow x}$ if and only if $\mathit{root}(T)\leq x$. That is, $r(T)\leq a$ if and only if $\mathit{root}(T)\leq a$ so that, 
indeed, $\mathit{root}(T)=r(T)$. It is clear that the map $\mathit{root}$ is surjective and thus the dual homomorphism 
$D\hookrightarrow K(D)$ is injective.
\end{proof}

\vspace{3mm}
\noindent
The following proposition will be used to obtain some important results in the next section of this paper.

\begin{prop}\label{KDorder}\label{prop:root=x}
Let $D$ be a finite distributive lattice and $X$ its dual poset. Let $S\in \mathcal{P}_r(X)$ and let $x\in X$. 
Identifying $\mathcal{P}_r(X)$ with $J(K(D))$ we have the following equivalences
\begin{align*}
\qquad & \mathit{root}(S)=x\\
                  \iff & S\preceq 1\nto x \mbox{ but } S\not\preceq 1\nto\kappa(x)\\
                  \iff &\mbox{ it is not the case that }( S\preceq 1\nto x \implies  S\preceq 1\nto\kappa(x) ).
\end{align*}
\end{prop}

\begin{proof}
We first assume that $S\preceq 1\nto x \mbox{ and } S\not\preceq 1\nto\kappa(x)$. Then 
$S\subseteq {\downarrow}x$ and $S\not\subseteq {\downarrow}\kappa(x)$. 
It follows that for each $s\in S$ we have $s\leq x$ and there is $t\in S$ with $t\not\leq \kappa(x)$. 
Therefore, by Proposition~\ref{cor1}, we have $x\leq t$. Since $t\in S$, we obtain $t=x$.  So $x\in S$. This implies that $x$ is the root of $S$, which means that 
$\mathit{root}(S)=x$. Conversely, suppose $\mathit{root}(S)=x$. Then $S\subseteq {\downarrow}x$ and $x\in S$. So  $S\preceq 1\nto x$.
On the other hand, we know that $y\not\leq \kappa(y)$, 
for each $y\in J(D)$. Therefore, $x\not\leq \kappa(x)$ and thus $S\not\subseteq {\downarrow}\kappa(x)$. 
This implies  that $S\not\preceq 1\nto \kappa(x)$. 
Finally, it is obvious that ($S\preceq 1\nto x \mbox{ and } S\not\preceq 1\nto\kappa(x)$) is equivalent to 
($\mbox{it is not the case that }( S\preceq  1\nto x \implies  S\preceq 1\nto\kappa(x))$). This finishes
the proof of the proposition. 
\end{proof}

\noindent
Since pre-Heyting algebras are the algebras for the functor $K$, we can construct free pre-Heyting 
algebras from the functor $K$ as we constructed them for free weak Heyting algebras. Given an 
order-preserving map $f: X \to X'$ between two finite posets $X$ and $X'$ we define 
${\mathcal P}_r(f):{\mathcal P}_r(X)\to {\mathcal P}_r(X')$ by setting ${\mathcal P}_r(f) = f[\ ]$.
It is easy to see that this is the action of the functor ${\mathcal P}_r$ dual to $K$. Then we will have 
the analogues of Theorems~\ref{freewHA} and \ref{freewHAdually} for free pre-Heyting algebras. 

We consider the following  sequence of distributive lattices:
\begin{align*}
D_0 \quad& = \quad  F_{DL}(n),\\
D_{k+1} & = \quad D_0 + K(D_k),\\
i_0\ :\ \ &D_0  \to \ D_0+K(D_0)\ =D_1 \mbox{ the embedding given by coproduct},\\ 
i_k\ :\ \ &D_k \to D_{k+1}  \mbox{ where } i_k=id_{D_0} +K(i_{k-1}).
\end{align*}

\noindent Similarly to weak Heyting algebras we have the following description of free pre-Heyting algebras.

\begin{thm}\label{freeqHA}
The direct limit $(D_\omega, (D_k\to D_\omega)_k)$ in ${\bf DL}$ of the system 
$(D_k, i_k:D_k\to D_{k+1})_k$ with the binary operation $\to_\omega:D_\omega\times D_\omega \to D_\omega$ 
defined by  $a \to_\omega b = a \to_k b$, for $a,b\in D_k$ is the free $n$-generated pre-Heyting algebra.
\end{thm}

\noindent

Let
$X_0$ be the dual of $D_0$ and let 
$$
X_{k+1} = X_0 \times \cl{P}_r(X_k)
$$
be the dual of $D_{k+1}$.
\begin{thm} The sequence $(X_k)_{k<\omega}$ with maps $\pi_k:X_0
  \times \cl{P}_r(X_k)\to X_k$ defined by
$$
\pi_k = id_{X_{0}}\times \mathcal{P}_r(\pi_{k-1})\ \ \mbox{i.e.}\ \ \pi_k(x,A)=(x,\pi_{k-1}[A])
$$
is dual to the sequence $(D_k)_{k<\omega}$ with maps $i_k: D_k \to
D_{k+1}$. In particular, the $\pi_k$'s are surjective. 
\end{thm}
\begin{proof}
The proof is analogues to the proof of Theorem~\ref{freewHAdually}.
\end{proof}

\section{Heyting algebras}\label{sec:HA}
In this section we will apply the technique of building free weak and pre-Heyting algebras to describe
free Heyting algebras. We recall the following definition of Heyting algebras relative to weak Heyting 
algebras. 

\begin{defi}\label{HA}\cite{Joh82}
A weak Heyting algebra $(A,\to)$ is called a {\em Heyting algebra}, {\em HA} for short, if 
the following two axioms are satisfied 
for all $a,b\in A$:

\begin{enumerate}[(7)]
\item  $b\leq a\to b$,

\item[(8)]  $a\wedge (a\to b)\leq b$.
\end{enumerate}
\end{defi}

\begin{lem}\label{FIC}
Let $(A,\to)$ be a weak Heyting algebra. Then $(A,\to)$ is a Heyting algebra iff 
for each $a\in A$ we have 
$$1\to a =a.$$
\end{lem}
\begin{proof}
Let $(A,\to )$ be a weak Heyting algebra satisfying $1\to a =a $ for each $a\in A$. Then for each $a,b\in A$ we have 
$b = 1\to b = (1\vee a) \to b$. By Definition~\ref{wHA}(3), $(1\vee a) \to b = (1\to b)\wedge (a\to b) = b\wedge (a\to b)$. So 
$b\leq a\to b$ and (7) is satisfied. Now $a \wedge (a\to b) = (1\to a)\wedge (a\to b)$. By Definition~\ref{wHA}(4),
$(1\to a)\wedge (a\to b)\leq 1\to b = b$. So $a\wedge (a\to b) \leq b$ and (8) is satisfied. 
Conversely, by (7) we have that $b\leq 1\to b$, and by (8), $1\to b = 1\wedge (1\to b)\leq b$.
\end{proof}

\noindent
Let $D$ be a finite distributive lattice. We have seen how to build the free weak 
Heyting algebra and the free pre-Heyting algebra over $D$ incrementally. 
Let $F_{HA}(D)$ denote the free HA freely generated by the  distributive 
lattice $D$. Further let $F_{HA}^n(D)$ denote the elements of $F_{HA}(D)$ of
$\to$-rank less than or equal to $n$. Then each $F_{HA}^n(D)$ is a distributive
lattice and the  lattice reduct of $F_{HA}(D)$ is the direct limit (union) of the chain
\[
D=F_{HA}^0(D)\subseteq F_{HA}^1(D)\subseteq F_{HA}^2(D)\ldots
\]
and the implication on $F_{HA}(D)$ is given by the maps 
$\to:(F_{HA}^n(D))^2\to F_{HA}^{n+1}(D)$ with $(a,b)\mapsto(a\to b)$.
Further, since any Heyting algebra is a pre-Heyting algebra and the inclusion 
$F_{HA}^n(D)\subseteq F_{HA}^{n+1}(D)$ may be seen as given by the mapping 
$a\mapsto (1\to a)$, the natural maps sending generators to generators make the 
following colimit diagrams commute
$$
\xymatrix@M=5pt{
D \ar@{->>}[d]^{id} \ar@{^(->}[r]^{\!\!\!\!\!i_0}   
	& K(D)   \ar@{->>}[d]^{g_1} \ar@{^(->}[r]^{\!\!\!i_1}   
	& K(K(D)) \ar@{->>}[d]^{g_2} \ar@{^(->}[r]^{\ \ \ i_2} 
	&\ldots \\
D\ar@{^(->}[r]^{\!\!\!\!\!j_0}                           
	& F_{HA}^1(D)\ar@{^(->}[r]^{\!\!\!j_1}                                 
	&F_{HA}^2(D)\ar@{^(->}[r]^{\ \ \ j_2} 
	&\ldots 
}
$$
Notice that under the assignment  $a\mapsto (1\to a)$, the equation (7) becomes  
$1\to b\leq a\to b$ which is true in any pHA by (2), and (8) becomes
$(1\to a)\wedge(a\to b)\leq(1\to b)$ which is true in any pHA by virtue of (4). 
So these equations are already satisfied in the steps of the upper sequence.
However, an easy calculation shows that for $D$ the three-element
lattice $1\to(u\to 0)$ and \mbox{$(1\to u)\to(1\to 0)$} are not equal (where $u$ is the 
middle element) thus the implication is not well-defined on the limit of the upper 
sequence. We remedy this by taking a quotient with respect to the relational schema corresponding 
to
\begin{enumerate}
\item[{\rm(9)}]
\ \ \ \ \ \ \ \ \ \ \ \ \ \ \ \ \ \ \ \ \ \ \ \ \ \ \ \ \ \ $1\to(a\to b)=(1\to a)\to(1\to b)$
\end{enumerate}
in the second iteration of the functor $K$ and onwards. This equation, which is a particularly 
trivial family of instances of Frege's axiom
$a\to(b\to c)=(a\to b)\to(a\to c)$ is of course true in intuitionistic logic but may not be of 
much interest there. From the algebraic perspective, this relational scheme is of course
very natural as it exactly ensures that the assignment $a\mapsto\ 1\to a$ preserves 
implication. As we will see subsequently, (9) will make sure that the embedding 
of each element of the constructed sequence in the next in an $\to$-homomorphism.
We proceed, as we've done throughout this paper by identifying the dual 
correspondent of this equation.

\begin{prop}
\label{prop:KKDleq}
Let $D$ be a finite distributive lattice and $X$ its dual poset. Further, let $\theta$ be a 
congruence on $K(K(D))$. Then the following are equivalent:
\begin{conditionsiii}
\item For all $a,b\in D$ the inequality $[(1\nto a)\nto(1\nto b)]\leq [1\nto(a\nto b)]$ holds 
in $K(K(D))/\theta$;
\item For all $x,y\in X$ the inequality $[(1\nto x)\nto(1\nto\kappa(y))]\leq [1\nto(x\nto\kappa(y))]$ 
holds in $K(K(D))/\theta$
\end{conditionsiii}
\end{prop}

\begin{proof}
(i) implies (ii) is clear since (ii) is a special case of (i). We prove that (ii) implies (i) (leaving out the 
brackets $[\ ]$ of the quotient map).
\begin{align*}
(1\nto a)\nto(1\nto b)&=(\bigvee_{X\ni x\leq a}(1\nto x))\nto(\bigwedge_{X\ni y\nleq b}(1\nto\kappa(y)))\\
                           &=\bigwedge_{\substack{X\ni x\leq a\\X\ni y\nleq b}}((1\nto x)\nto(1\nto\kappa(y)))\\
                           &\leq \bigwedge_{\substack{X\ni x\leq a\\X\ni y\nleq b}}(1\nto(x\nto\kappa(y)))\\
                           &= 1\nto\bigwedge_{\substack{X\ni x\leq a\\X\ni y\nleq b}}(x\nto\kappa(y))\\
                           &= 1\nto(a\nto b)).
\end{align*}
\end{proof}
We are now ready to translate this into a dual property which we will call $(G)$ after Ghilardi 
who introduced it in \cite{Ghilardi92}.

\begin{prop}\label{VP:cor}
Let $X_0$ be a finite poset, $X_1$ a subposet of ${\mathcal P}_r(X_0)$, and $X_2$ a subposet of 
${\mathcal P}_r(X_1)$. For $i=1,2$, let $(\star)_i$ be the condition
$$
x\in T\in X_i\ \implies\ T_x=({\downarrow}x\cap T)\in X_i
$$
If the conditions $(\star)_1$ and $(\star)_2$ both hold then the following are equivalent:
\begin{conditionsiii}
\item $\forall x,y\in X_0$ the inequalities 
\[ 
[(1\nto x)\nto(1\nto\kappa(y))]\leq[1\nto(x\nto\kappa(y))]\mbox{ hold in }{\mathcal O}(X_2);
\] 

\item $\forall \tau\in X_2 \  \forall T\in \tau\ \forall S\in X_1$
\[
(S\leq T\ \implies\ \exists T'\in\tau\ (T'\leq T\mbox{ and } \mathit{root}(S)=\mathit{root}(T')). \ \ \ \ (G)
\]
\end{conditionsiii}
\end{prop}

\begin{proof}
First we prove that (i) implies (ii). To this end suppose (i) holds and let 
$T\in\tau\in X_2$ and $S\in X_1$. Suppose that for all $T'\in\tau$ either 
$T'\nleq T$ or $\mathit{root}(S)\neq \mathit{root}(T')$. Now consider $\tau_T={\downarrow}T\cap\tau$.  
By $(\star)_2$ we have that $\tau_T\in X_2$.
Then we have $\mathit{root}(S)\neq \mathit{root}(T')$ for all 
$T'\in\tau_T$. Thus, letting $x=\mathit{root}(S)$ and using the observation in 
Proposition~\ref{KDorder}, we obtain 
\begin{align*}
&\forall T'\in\tau_T\qquad\qquad \mathit{root}(T')\neq x\\
\iff\ &\forall T'\in\tau_T\quad (T'\preceq 1\nto x\ \implies\ T'\preceq 1\nto\kappa(x))\\
\iff\ &  \tau_T\preceq (1\nto x)\nto(1\nto\kappa(x))\\
\implies\ &  \tau_T\preceq 1\nto(x\nto\kappa(x))\\
\iff\ & \forall T'\in\tau_T\qquad T'\preceq x\nto\kappa(x)\\
\iff\ & \forall T'\in\tau_T\ \forall y\in T'\quad(y\leq x\ \implies\ y\leq\kappa(x))\\
\iff\ & \forall T'\in\tau_T\qquad x\not\in T'\\
\implies\ & \qquad x\not\in T.
\end{align*}
The two implications come from the fact that we assume that (i) holds and 
because, in particular, $T\in\tau_T$. Now we have $x=\mathit{root}(S)\in S$
but $x\not\in T$ so $S\nleq T$ and we have proved (ii) by contraposition.

Now suppose (ii) holds, let $x,y\in X_0$, and let $\tau\in X_2$ such that 
$\tau\preceq (1\nto x)\nto(1\nto\kappa(y))$. Then
\begin{align*}
&\tau\preceq (1\nto x)\nto(1\nto\kappa(y))\\
\iff\ &\forall T\in\tau\quad (T\preceq 1\nto x\ \implies\ T\preceq 1\nto\kappa(y))\\
\iff\ &\forall T\in\tau\quad (T\leq{\downarrow}x\ \implies\ T\leq{\downarrow}\kappa(y)).
\end{align*}
We want to show that $T\preceq x\nto\kappa(y)$ for each $T\in\tau$. That is, that for all $z\in T$
we have $z\leq x$ implies $z\leq\kappa(y)$. So let $z\in T$ with $z\leq x$. 
By  $(\star)_1$ we have that $ T_z = {\downarrow}z\cap T\in X_1$.
Since $T_z\leq T$ it follows by (ii) that
\[
\exists T'\in\tau\qquad (T'\leq T\mbox{ and } z=\mathit{root}(T_z)=\mathit{root}(T')).
\]
Now $x\geq z=\mathit{root}(T_z)=\mathit{root}(T')$ implies that $T'\leq{\downarrow}x$ and thus we have 
$T'\leq{\downarrow}\kappa(y)$. In particular, $z=\mathit{root}(T')\leq\kappa(y)$. That is, we 
have shown that for all $z\in T$, if $z\leq x$ then $z\leq\kappa(y)$ as required. 
\end{proof}

Our strategy in building the free $n$-generated Heyting algebra will be to start with $D$, the free 
$n$-generated distributive lattice, embed it in $K(D)$, and then this in a quotient of $K(K(D))$
obtained by quotienting out by $1\nto(a\nto b)=(1\nto a)\nto(1\nto b)$ for $a,b\in D$. For the further 
iterations of $K$ this identification is iterated. The following is the general situation that we need 
to consider, viewed dually: 
$$
\xymatrix@M=5pt{
X_0\ \ar@{<<-}[r]^{\!\!\!\!\!\!\mathit{root}} \ &{\mathcal P}_r(X_0)\ \ar@{<<-}[r]^{\mathit{root}} \   & {\mathcal P}_r({\mathcal P}_r(X_0)) \\
                                                 &X_1\ \ar@{_(->}[u]  \ \ar@{<<-}[r]^{\mathit{root}}        & {\mathcal P}_r(X_1) \ar@{_(->}[u]\\
                                                 &                                                                       & X_2          \ar@{_(->}[u]
}
$$
 We now consider the following sequence of finite posets
\begin{align*}
X_0 & =J(F_{DL}(n))(={\mathcal P}(n))\\
X_1 & ={\mathcal P}_r(X_0)\\
\hspace{-.6cm}\mbox{for }n\geq 1\quad X_{n+1} & =\{\tau\in{\mathcal P}_r(X_n)\mid 
             \forall T\in \tau\ \forall S\in X_n\ \\
       &\qquad\quad (S\leq T\ \implies\ \exists T'\in\tau\ (T'\leq T\mbox{ and } \mathit{root}(S)=\mathit{root}(T'))\}.
\end{align*}
We denote by $\nabla$ the sequence 
$$
\xymatrix@M=5pt{
X_0\ \ar@{<-}[r]^{\mathit{root}} \ & X_1\ \ar@{<-}[r]^{\mathit{root}} \   & X_2 \ \dots
                                              }
$$
For $n\geq 1$, we say that $\nabla$ satisfies $(\star)_n$ if 
$$
x\in T\in X_n\ \implies\ T_x=({\downarrow}x\cap T)\in X_n.
$$

\begin{lem}\label{Ro:lem}
$\nabla$ satisfies $(\star)_n$ for each $n\geq 1$ and the root maps
$\mathit{root}: X_{n+1}\to X_n$ are surjective for each $n\geq 0$.
\end{lem}
\begin{proof}
$X_1$ consists of all rooted subsets of $X_0$ and thus $(\star)_1$ is clearly 
satisfied. Now let $n\geq 2$. We assume that $T\in \tau\in X_{n}$ and we show that 
$\tau_T = {\downarrow} T\cap \tau$ also belongs to $X_{n}$. So let $U\in \tau_T$, 
$S\in X_n$ and $S\leq U$. Then since $U\in \tau$, there exists $U'\in \tau$ such that 
$U'\leq U$ and $\mathit{root}(S)=\mathit{root}(U')$. But $U\in \tau_T$ implies that $U\leq T$. Therefore 
we have $U'\leq T$ and so $U'\in \tau_T$. Thus, $\tau_T\in X_{n}$ and $\nabla$ 
satisfies $(\star)_n$, for each $n\geq 1$.
Finally, we show that all the $\mathit{root}$ maps are surjective. To see this, assume $U\in X_n$. We 
show that ${\downarrow}U\in X_{n+1}$. Suppose $T\in {\downarrow}U$ and for some 
$S\in X_n$ we have $S\leq T$. Then $S\in {\downarrow}U$ and by setting $T'=S$ we 
easily satisfy the condition $(G)$. Finally, note that $\mathit{root}({\downarrow}U) = U$ and thus 
$\mathit{root}: X_{n+1}\to X_n$ is surjective. 
\end{proof}

Let $\Delta$ be the system
$$
\xymatrix@M=5pt{
D_0\ \ar@{^(->}[r]^{i_0} \ & D_1\ \ar@{^(->}[r]^{i_1} \   & D_2 \ \dots
                                              }
$$
of distributive lattices dual to $\nabla$. For each $n\geq0$, $i_n:D_n\to D_{n+1}$,  
is a lattice homomorphism dual to $\mathit{root}$. By Theorem~\ref{thrm:JKD}(iii) 
$i_n(a)= [1\nto a]_\approxeq$, for $a\in D_n$. By Lemma~\ref{Ro:lem}, $\mathit{root}$ is 
surjective, so each $i_n$ is injective. Each $X_{n+1}\subseteq {\mathcal P}_r(X_n)$ 
so that each $D_{n+1}$ is a quotient of $K(D_n)$ and thus, for each $n$, we also have 
implication operations:
\begin{align*}
\to_n:\ D_n\times D_n &\to \ D_{n+1}\\
		   (a,b)\ \   &\mapsto \ [a\nto b].
\end{align*}
Here $[a\nto b]$ is the equivalence class of $a\nto b$ as an element in $D_{n+1}$.
Let $D_\omega$ be the limit of $\Delta$ in the category of distributive lattices then 
$D_\omega$ is naturally turned into a Heyting algebra.

\begin{lem}\label{Vo:lem}
The operations $\to_n$ can be extended to an operation $\to_\omega$ on $D_\omega$
and the algebra $(D_\omega, \to_\omega)$ is a Heyting algebra.
\end{lem}
\begin{proof}
The colimit $D_\omega$ of $\Delta$ may be constructed as the union of the $D_n$s with 
$D_n$ identified with the image of $i_n:D_n\hookrightarrow D_{n+1}$. It is then clear that
the 
operations $\to_n:\ D_n\times D_n \to \ D_{n+1}$ yield a total, well-defined 
binary operation on the limit $D_\omega$ provided, for all $n\in \omega$ and all $a,b\in D_n$, we have 
$i_{n+1}(\to_{n}(a,b))=\to_{n+1}(i_{n}(a),i_{n}(b))$. But this is exactly
\[
1\to_{n+1}(a\to_nb)=(1\to_na)\to_{n+1}(1\to_nb).
\]
As we've shown in Corollary~\ref{VP:cor} and Lemma~\ref{Ro:lem}, the sequence $\nabla$,
and thus the dual sequence $\Delta$ have been defined exactly so that this holds. It remains 
to show that the algebra $(D_\omega,\to)$ is a Heyting algebra. Let $a\in D_\omega$, then there is some
$n\geq 0$ with $a\in D_n$. Now $a\to_\omega a=a\to_n a\in D_{n+1}$. Since $D_{n+1}$ is a 
further quotient of $K(D_n)$ and $a\to_n a=1$ already in $K(D_n)$, this is certainly also true
in $D_{n+1}$ and $1_{D_{n+1}}=1_{D_\omega}$ so the equation (1) of weak Heyting algebras
is satisfied in $(D_\omega, \to_\omega)$. Similarly each of the equations (2)--(4) are satisfied
in $(D_\omega, \to_\omega)$ so that it is a weak Heyting algebra. 
Finally, note that as $i_n(a)= 1\to_n a$ for each $a\in D_n$ and $n\in \omega$, we have 
$a = 1\to_\omega a$ for each $a\in D_\omega$. By Lemma~\ref{FIC}, this implies that $(D_\omega,\to_\omega)$ is a Heyting algebra.
\end{proof}

\begin{cor}
For each $n\in \omega$, if $D_0$ is the $n$-generated free distributive lattice, then 
$(D_\omega, \to_\omega)$ constructed above is the  $n$-generated free Heyting algebra. 
\end{cor}
\begin{proof}
Let $F_{HA}(n)$ denote the free HA freely generated by $n$ generators. This is of course
the same as the free HA generated by $D$, $F_{HA}(D)$, where $D$ is the free distributive
lattice generated by $n$ elements. As discussed at the beginning of this section this lattice
is the colimit (union) of the chain
\[
D=F_{HA}^0(D)\subseteq F_{HA}^1(D)\subseteq F_{HA}^2(D)\ldots
\]
and the implication on $F_{HA}(D)$ is given by the maps 
$\to:(F_{HA}^n(D))^2\to F_{HA}^{n+1}(D)$ with $(a,b)\mapsto(a\to b)$. Further, the natural 
maps sending generators to generators make the following colimit diagrams commute
$$
\xymatrix@M=5pt{
D \ar@{->>}[d]^{id} \ar@{^(->}[r]   
	& K(D)   \ar@{->>}[d] \ar@{^(->}[r]  
	& K(K(D)) \ar@{->>}[d] \ar@{^(->}[r]
	&\ldots \\
D\ar@{^(->}[r]^{\!\!\!\!\!j_0}                           
	& F_{HA}^1(D)\ar@{^(->}[r]^{\!\!\!j_1}                                 
	&F_{HA}^2(D)\ar@{^(->}[r]^{\ \ \ j_2} 
	&\ldots 
}
$$
Now, the system $\Delta$ is obtained from the upper sequence by quotienting out by the 
equations $1\to_{n+1}(a\to_nb)=(1\to_na)\to_{n+1}(1\to_nb)$ for each $n\geq 0$. Since 
these equations all hold for the lower sequence, it follows that the $D_n$s are intermediate 
quotients:
$$
\xymatrix@M=5pt{
D \ar@{->>}[d]^{id} \ar@{^(->}[r]  
	& K(D)   \ar@{->>}[d] \ar@{^(->}[r]  
	& K(K(D)) \ar@{->>}[d] \ar@{^(->}[r]
	&\ldots \\
D_0 \ar@{->>}[d]^{id} \ar@{^(->}[r]^{\!\!\!\!\!i_0}   
	& D_1   \ar@{->>}[d] \ar@{^(->}[r]^{\!\!\!i_1}   
	& D_2 \ar@{->>}[d] \ar@{^(->}[r]^{\ \ \ i_2} 
	&\ldots \\
D\ar@{^(->}[r]^{\!\!\!\!\!j_0}                           
	& F_{HA}^1(D)\ar@{^(->}[r]^{\!\!\!j_1}                                 
	&F_{HA}^2(D)\ar@{^(->}[r]^{\ \ \ j_2} 
	&\ldots 
}
$$
Therefore, $F_{HA}(n)$ is a homomorphic image of $D_\omega$. Moreover, any map $f:n\to B$ with
$B$ a Heyting algebra defines a unique extension $\tilde{f}:F_{HA}(n)\to B$ such that 
$\tilde{f}\circ i=f$, where $i:n\to F_{HA}(n)$ is the injection of the free generators. Since $i$ 
actually maps into the sublattice of $F_{HA}(n)$ generated by $n$, which is the initial lattice
$D=D_0$ in our sequences, composition of $\tilde{f}$ with the quotient map from $D_\omega$ to 
$F_{HA}(n)$ shows that $D_\omega$ also has the universal mapping property (without the 
uniqueness). The uniqueness follows since $D_\omega$ clearly is generated by $n$ as HA 
(since $D_0$ is generated by $n$ as a lattice, $D_1$ is generated by $D_0$ using
$\to_0$, and so on). Since the free HA on $n$ generators is unique up to isomorphism and 
$D_\omega$ has its universal mapping property and is a Heyting algebra, it follows it is the 
free HA (and the quotient map from $D_\omega$ to $F_{HA}(n)$ is in fact an isomorphism).
\end{proof}

\section{A coalgebraic representation of wHAs and PHAs}\label{top:sec}

In this section we discuss coalgebraic semantics for weak and pre-Heyting algebras. 
A coalgebraic representation of modal algebras and distributive modal algebras can be found in 
\cite{Abramsky05}, \cite{KKV04} and \cite{Palmigiano04}, \cite{bezh-kurz:calco07}, respectively.

We recall that a {\it Stone space} is a compact Hausdorff space with a basis of clopen sets. 
For a Stone space $X$, its {\it Vietoris space} $V(X)$ is defined as
the set of all closed subsets of $X$, endowed with the topology generated by
the subbasis 

\begin{enumerate}[(1)]
\item $\Box U =\{F\in V(X):F\subseteq U \}$,

\item $\Diamond U =\{F\in V(X):F\cap U\neq \emptyset\}$,
\end{enumerate}
where $U$ ranges over all clopen subsets of $X$. It is well known that $X$ is a Stone space iff 
$V(X)$ is a Stone space. 
Let $X$ and $X'$ be Stone spaces and $f:X \to X'$ be a continuous map. Then $V(f)= f[\ ]$ is a continuous map between
$V(X)$ and $V(X')$. We denote by $V$ the functor on Stone spaces that maps every Stone space $X$ to its Vietoris 
space $V(X)$ and  maps every continuous map $f$ to $V(f)$. 
Every modal algebra $(B,\Box)$ corresponds to a coalgebra $(X, \alpha: X \to VX)$ 
for the Vietoris functor on Stone spaces \cite{Abramsky05,KKV04}. Coalgebras corresponding to distributive modal algebras 
are described in \cite{Palmigiano04} and \cite{bezh-kurz:calco07}. We note that modal algebras as well as distributive modal algebras
are given by rank 1 axioms. Using the same technique as in Section 3, one can obtain 
a description of free modal algebras and free distributive modal algebras
\cite{Abramsky05}, \cite{Ghilardi95}, \cite{bezh-kurz:calco07}. 

Our goal is to characterize coalgebras corresponding to  weak Heyting algebras and pre-Heyting algebras. 
Recall that a {\it Priestley space} is a pair $(X,\leq)$ where $X$ is a Stone space and $\leq$ is a reflexive, antisymmetric and transitive relation 
satisfying the {\it Priestley separation axiom}:
\begin{center}
If $x,y\in X$ are such that $x\not\leq y$, then there exists a clopen downset $U$\\
 with $y\in U$ and $x\notin U$.
\end{center}

\noindent We denote by {\bf PS} the category of Priestley spaces and order-preserving continuous maps. 
It is well known that every distributive lattice $D$ can be represented as a lattice of all clopen downsets of the Priestley space of 
its prime filters. Given a Priestley space $X$, let $V_r(X)$ be a subspace of $V(X)$ of all closed rooted subsets of $X$.
The same proof as for $V(X)$ shows that $V_r(X)$ is a Stone space.

\begin{lem}
Let $X$ be a Priestley space. Then 
\begin{enumerate}[\em(1)]
\item $(V(X), \subseteq)$ is a Priestley space.
\item $(V_r(X), \subseteq)$ is a Priestley space. 
\end{enumerate}
\end{lem}
\begin{proof}
(1) As we mentioned above $V(X)$ is a Stone space. Let $F,F'\in V(X)$ and $F\not\subseteq F'$. Then there exists $x\in F$ such that 
$x\notin F'$. Since every compact Hausdorff space is normal, there exists a clopen set $U$ 
such that $F'\subseteq U$ and $x\notin U$. Thus, $F'\in \Box U$ and $F\notin \Box U$. All we need to observe now is that 
for each clopen $U$ of $X$, the set $\Box U$ is a clopen $\subseteq$-downset of $V(X)$. But this is obvious.

   The proof of (2) is the same as for (1). 
\end{proof}

\noindent Let $(X,\leq)$ and $(X',\leq')$ be Priestley spaces and $f:X \to X'$ a continuous order-preserving map. Then it is easy to check 
that $V(f)= f[\ ]$ is a continuous order-preserving map between $(V(X),\subseteq)$ and $(V(X'),\subseteq)$, 
and $V_r(f)= f[\ ]$ is a continuous order-preserving map between $(V_r(X),\subseteq)$ and $(V_r(X'),\subseteq)$.
Thus, $V$ and $V_r$ define functors on the category of Priestley spaces.

\begin{defi}\label{wHS}(Celani and Jansana \cite{CJ05})
A {\it weak Heyting space} is a triple $(X,\leq,R)$ such that $(X,\leq)$ is a Priestley space 
and $R$ is a binary relation on $X$ satisfying the following conditions:
\begin{enumerate}[(1)]
\item $R(x)=\{y\in X: xRy\}$ is a closed set, for each $x\in X$.
\item For each $x,y,z\in X$ if $x\leq y$ and $xRz$, then $yRz$. 
\item For each clopen set $U\subseteq X$ the sets $[R](U)=\{x\in X: R(x)\subseteq U\}$ and $\langle R\rangle (U)=
\{x\in U: R(x)\cap U\neq \emptyset\}$ are clopen.
\end{enumerate}
\end{defi}

\noindent Let $(X,\leq, R)$ and $(X',\leq', R')$ be two weak Heyting
spaces. We say that $f:X \to X'$ is a weak Heyting morphism if $f$ is
continuous, $\leq$-preserving and $R$-bounded morphism (i.e., for each
$x\in X$ we have $fR(x)=R'f(x)$).  Then the category of weak Heyting
algebras is dually equivalent to the category of weak Heyting spaces
and weak Heyting morphisms \cite{CJ05}.  We will quickly recall how
the dual functors are defined on objects. Given a weak Heyting algebra
$(A,\to)$ we take a Priestley dual $X_A$ of $A$ and define $R_A$ on
$X_A$ by setting: for each $x,y\in X_A$, $xR_Ay$ if for each $a,b\in
A$, $a\to b\in x$ and $b\in x$ imply $b\in y$. Conversely, if
$(X,\leq, R)$ is a weak Heyting space, then we take the distributive
lattice of all clopen downsets of $X$ and for clopen downsets
$U,V\subseteq X$ we define $U \to V = \{x\in X: R(x)\cap U\subseteq
V\}$.

\begin{rem}
In fact, Celani and Jansana  \cite{CJ05} work with clopen upsets instead of downsets and the inverse of the relation $R$. 
We chose working with downsets to be consistent with the previous parts of this paper. 
\end{rem}

\begin{thm}\label{wHcoalg}
The category of weak Heyting spaces is isomorphic to the category of Vietoris coalgebras on the category of Priestley spaces.
\end{thm}
\begin{proof}
Given a weak Heyting space $(X,\leq, R)$. We consider a coalgebra $(X, R(.):X \to V(X))$. The map $R(.)$ is well defined by 
Definition~\ref{wHS}(1). It is order-preserving by Definition~\ref{wHS}(2) and is continuous by Definition~\ref{wHS}(3). 
Thus, $(X, R(.):X \to V(X))$ is a $V$-coalgebra. Conversely, let $(X,\alpha: X \to V(X))$ be a $V$-coalgebra. 
Then $(X,R_\alpha)$, where $x R_\alpha y$ iff $y\in \alpha(x)$, is a weak Heyting space. Indeed, 
$R$ being well defined and order-preserving imply conditions (1) and (2) of Definition~\ref{wHS}, respectively.
Finally, $\alpha$ being continuous implies condition (3) of Definition~\ref{wHS}. 
That this correspondence can be lifted to the isomorphism of categories is easy to check. 
\end{proof}

We say that a weakly Heyting space $(X,\leq,R)$ is a {\it pre-Heyting space} if for each $x\in X$ the set $R(x)$ is rooted.

\begin{thm}
\quad
\begin{enumerate}[\em(1)]
\item The category of pre-Heyting algebras is dually equivalent to the category of pre-Heyting spaces.
\item The category of pre-Heyting spaces is isomorphic to the category of \ $V_r$-coalgebras on the category 
of Priestley spaces. 
\end{enumerate}
\end{thm}
\begin{proof}
(1) By the duality of weak Heyting algebras and weak Heyting spaces it is sufficient to show  
that a weak Heyting algebra satisfies conditions (5)--(6) 
of Definition~\ref{QHA} iff $R(x)$ is rooted. We will show, as in Theorem~\ref{thrm:JKD}, that the axiom (5) is equivalent to 
$R(x)\neq \emptyset$, for each $x\in X$, while axiom (6) is equivalent to $R(x)$ having a unique maximal element.
Assume that a weak Heyting space $(X,\leq,R)$ validates axiom (5). Then in the weak Heyting algebra of all clopen downsets of $X$
we have $X\to \emptyset = \emptyset$. 
Thus for each $x\in X$ we have $R(x)\subseteq \emptyset$ iff $x\in \emptyset$. Thus, for each $x\in X$ 
we have $R(x)\neq \emptyset$. Now suppose for each clopen downsets $U,V\subseteq X$ the following holds 
$X\to (U \cup V) \subseteq (X\to U) \cup (X\to V)$. Then we have that $R(x)\subseteq U\cup V$ implies $R(x)\subseteq U$ or 
$R(x)\subseteq V$. Since $R(x)$ is closed and $X$ is a Priestley space, we have that every point of $R(x)$ is below
some maximal point of $R(x)$. We assume that there exists more than one maximal point of $R(x)$. Then the same argument as in  
\cite[Theorem 2.7(a)]{BezhBezh08} shows that there are clopen downsets $U$ and $V$ such that $R(x)\subseteq U\cup V$, but 
$R(x) \not\subseteq U$, $R(x)\not\subseteq V$. This is a contradiction, so $R(x)$ is rooted.
On the other hand, it is easy to check that if $R(x)$ is rooted for each $x\in X$, then (5) and (6) are valid. 
Finally, a routine verification shows that this correspondence can be lifted to 
an isomorphism of the categories of 
pre-Heyting algebras and pre-Heyting spaces.

 The proof of (2) is similar to the proof of Theorem~\ref{wHcoalg}. The extra condition on pre-Heyting spaces obviously implies that 
a map $R(.): X \to V_r(X)$ is well defined and conversely $(X,\alpha: X\to V_r(X))$ being a coalgebra implies that 
$R_\alpha(x)$ is rooted for each $x\in X$. The rest of the proof is a routine check. 
\end{proof}

Thus, we obtained a coalgebraic semantics/representation of weak and pre-Heyting algebras.

\section{Conclusions and future work}
In this paper we described finitely generated free (weak, pre-) Heyting algebras 
using an initial algebra-like construction. The main idea is to split 
the axiomatization of Heyting algebras into its rank 1 and non-rank 1 parts. 
The rank 1 approximants of Heyting algebras are weak and pre-Heyting algebras. 
For weak and pre-Heyting algebras we applied the standard initial algebra construction and then adjusted it for 
Heyting algebras. We used Birkhoff duality for finite distributive lattices and finite posets to obtain the dual characterization 
of the finite posets that approximate the duals  of  free algebras. As a result, we obtained Ghilardi's 
representation of these posets in a systematic and modular way. We also gave a coalgebraic representation of weak and pre-Heyting algebras.

There are a few possible directions for further research.  
As we mentioned in the introduction, although we considered Heyting algebras (intuitionistic logic),
this method could be applied to other non-classical logics. More precisely, the method is available  
if a signature of the algebras for this logic can be obtained by adding an extra 
operator to a locally finite variety. Thus, various non-rank 1 modal logics such as ${\bf S4}$, ${\bf K4}$ and other
more complicated modal logics, as well as distributive modal logics, are the obvious candidates.
On the other hand, one cannot always expect to have such a simple representation of free algebras. 
The algebras corresponding to other many-valued logics such as $MV$-algebras, $l$-groups, $BCK$-algebras and so on, 
are other examples where this method could lead to interesting representations. 
The recent work \cite{BCG06} that connects ontologies with free distributive algebras with operators shows that 
such representations of free algebras are not only interesting from a theoretical point of view, but could have 
very concrete applications.

\vspace{2mm}

\noindent
{\bf Acknowledgements}
The first listed author would like to thank Mamuka Jibladze and Dito Pataraia for many interesting 
discussions on the subject of the paper. The authors are also very grateful to the anonymous referees
for many helpful suggestions that substantially improved the presentation of the paper. 

 \bibliographystyle{abbrv}
 \bibliography{MyBib}

\end{document}